\documentclass[10pt]{article} 
\usepackage[preprint]{rlc}

\usepackage{amssymb}            
\usepackage{mathtools}          
\usepackage{mathrsfs}           
\usepackage{graphicx}           
\usepackage{subcaption}         
\usepackage[space]{grffile}     
\usepackage{url}                
\usepackage{soul}

\usepackage{float}
\usepackage{listings}
\DeclareCaptionStyle{ruled}{labelfont=normalfont,labelsep=colon,strut=off} 
\floatstyle{ruled}
\newfloat{listing}{tb}{lst}{}
\floatname{listing}{\textbf{Game}}

\usepackage{amsthm}
\usepackage{comment}
\usepackage{algorithm}
\usepackage{algpseudocode}
\usepackage{booktabs}
\usepackage{makecell}
\usepackage{tikz}
\usetikzlibrary{shapes.geometric, arrows, calc, positioning}

\newtheorem{assumption}{Assumption}
\newtheorem{definition}{Definition}
\newtheorem{theorem}{Theorem}
\newtheorem{proposition}{Proposition}

\newtheorem{remark}{Remark}
\newtheorem{example}{Example}

\definecolor{rougeINRIA}{RGB}{230, 51, 18}
\definecolor{grisbleuINRIA}{RGB}{56,66,87}
\definecolor{orangeINRIA}{RGB}{240, 126, 38}
\definecolor{lilasINRIA}{RGB}{155,0,79}
\definecolor{bleuINRIA}{RGB}{20,136,202}
\definecolor{vertINRIA}{RGB}{149,193,31}
\definecolor{jauneINRIA}{RGB}{255,205,28}
\definecolor{mauveINRIA}{RGB}{101,97,169}
\definecolor{bleuclairINRIA}{RGB}{137,204,202}
\definecolor{vertclairINRIA}{RGB}{199,214,79}

\title{Steering Noncooperative Games Through Conjecture Design}

\author{Francesco Morri\thanks{Corresponding Author: francesco.morri@inria.fr}\\
    Inria, Univ. Lille\\
    CNRS, Centrale Lille\\[0.5ex]
    Department of Computer and Software Engineering\\
    Polytechnique Montreal
    \And
    Hélène Le Cadre \\
    Inria, Univ. Lille\\
    CNRS, Centrale Lille
    \And
    David Salas\\
    Instituto de Ciencias de la Ingeniería\\
    Universidad de O'Higgins
    \And
    Didier Aussel\\
    Laboratoire PROMESUPR 8521\\
    Université de Perpignan, Tecnosud}

\begin{document}

\maketitle

\begin{abstract}
In dynamic noncooperative games, each player makes conjectures about other players' reactions before choosing a strategy. However, resulting equilibria may be multiple and do not always lead to desirable outcomes. These issues are typically addressed separately — for example, through opponent modelling and incentive design. Drawing inspiration from conjectural variations games, we propose an incentive design framework in which a coordinator first computes an equilibrium by optimizing a predefined objective function, then communicates this equilibrium as a target for the players to reach. In a centralized setting, the coordinator also optimizes the conjectures to steer the players towards the target. In decentralized settings, players independently compute conjectures and update their strategies based on individual targets. We provide a guarantee of equilibrium existence in both cases. This framework uses conjectures not only to guide the system towards desirable outcomes but also to decouple the game into independent optimization problems, enabling efficient computation and parallelization in large-scale settings. We illustrate our theoretical results on classical representative noncooperative games, demonstrating its application potential.
\end{abstract}

\section{Introduction}
Achieving specific equilibria in multi-agent games remains a fundamental challenge, particularly when scaling to complex, high-dimensional environments. Existing solutions either reshape the game (mechanism design), reshape utilities (incentive design), or require online opponent inference. We introduce a conjecture-based modeling approach that synthesizes concepts from the three areas: our framework is grounded in conjectural variations games \citep{figuieres_theory_2004}, where each agent forms structured beliefs, called \textit{conjectures}, about the other players' behaviors and optimizes its strategy accordingly.
This generalization yields a model that is both computationally tractable and broadly applicable across diverse strategic settings.


Mechanism design has long provided a principled way to align agent objectives in structured environments such as markets and auctions \citep{myerson_optimal_1981, daskalakis_complexity_2014}. Traditionally, a central authority sets rules to elicit desirable behaviours. However, as highlighted in \citep{roughgarden_approximately_2019}, the complexity of these mechanisms often limits their practical deployment. In contrast, our framework avoids restrictive assumptions and aims for broad applicability, supporting arbitrary objective functions with limited central coordination.


Incentive design offers an alternative to mechanism design by directly shaping agents' utility functions rather than altering the game structure itself \citep{paccagnan_utility_2022}. Games often lead to multiple (possibly infinite) equilibria \citep{harsanyi_general_1988}, in these settings incentives can guide agents towards a desirable outcome. 
In \citep{ratliff_adaptive_2021}, the authors propose an algorithm for a coordinator adaptively designing incentives by learning the agents' decision-making process (decided by their types), in order to obtain a result that optimize its objective function. For Markov potential games, \citep{mguni_coordinating_2019} addresses the design of incentives that alter players’ reward functions, with theoretical guarantees on optimality.
Recently, \citep{huang_cost_2024} explores the problem of finding the minimum incentive needed to make players change equilibrium, studying bimatrix games.
Our framework generalizes this approach by allowing arbitrary game structures and agents' behaviors.
Finally, \citep{leonardos_exploration-exploitation_2021} examines how to reach a unique equilibrium in large cooperative games using Q-learning by adjusting the learning rate.
While their work frames this as an equilibrium selection problem, we categorize it under incentive design, as selecting a particular equilibrium typically requires first enumerating the set of possible equilibria. Unlike their cooperative setting, our framework assumes no cooperation among agents.

In large-scale multi-agent games, the non-stationarity arising from agents' mutual adaptation poses significant challenges to reaching stable strategies. Opponent modeling addresses this by explicitly incorporating predictions about other players' behavior, into each agent’s learning process \citep{albrecht_autonomous_2018}.
In algorithms such as LOLA \citep{foerster_learning_2018}, an agent optimizes its return under one step look-ahead of opponent learning, assuming access to either the opponents parameters or their actions.
Recent work \citep{duque_advantage_2025} shows that all opponent modeling algorithms can be reduced to a common rule aligning agents, thus assuming some underlying implicit coordination among the agents. In contrast, our approach externalizes the modeling process: agents optimize their utility functions based on fixed conjectures designed a priori, without requiring information exchange or behavioral assumptions during the course of the game.

Our main contributions can be summarized as follows:
\begin{itemize}
    \item We introduce a conjecture-based framework for $N$-player games that generalizes classical formulations by relaxing consistency constraints and decoupling conjecture design from strategy update.
    \item We prove the existence of consistent conjectures under standard regularity conditions and demonstrate that the resulting equilibria align with designer-specified objectives.
    \item We demonstrate the practical relevance of our framework by applying it to control and coordination games, where it enables players to achieve more efficient outcomes than without conjectures. We consider four games: two games involving two players that can be solved analytically, a large-scale coordination problem, and a saddle game.
\end{itemize}

\paragraph{Notations}
Throughout the paper we use $f(x;\theta)$ to indicate a function with variable $x$, parameterized by $\theta$. When $x = (x_i)_{i\in\mathcal{N}}$, $\nabla_{i} f(x;\theta)$ indicates the partial derivative of $f$ with respect to $x_i,\,\forall i \in \mathcal{N}$, and $\nabla$ represents the total derivative of $f$ with respect to $x$. Further, $\nabla_{-i}f:=\text{col}(\nabla_{j} f)_{j\neq i}$ refers to the stack of the partial derivatives of $f$ with respect to $x_j$ for all $j \in \mathcal{N} \setminus \{i\}$. $f \in \mathcal{C}^k$ is a function of smoothness at least $k$.

\section{Conjectural Models}
We consider a set $\mathcal{N} := \{1,\hdots,N\}$ of $N$ strategic players, each with decision variable $x_i\in\mathcal{X}_i=\mathbb{R}^{m_i}$,\,$m_i\in\mathbb{N}^\star,\,m:=\sum_{i\in\mathcal{N}} m_i$, and objective function $J_i:\mathbb{R}^m \rightarrow \mathbb{R}$. Player $i$'s objective function $J_i$ is differentiable in $x_1,\hdots,x_N$. We set $x_{-i}:=(x_j)_{j\in \mathcal{N}\setminus\{i\}}$ and define $x:=(x_i)_{i\in\mathcal{N}} \in \mathcal{X}:=\prod_{i\in\mathcal{N}}\mathcal{X}_i$ as the $N$ players' joint strategy profile. The interactions between the $N$ players give rise to a noncooperative game $$\mathcal{G}:=(\mathcal{N},\mathcal{X},(J_i)_{i\in\mathcal{N}}).$$

We aim to formulate an incentive design problem, where an $(N+1)$-player can coordinate the $N$ strategic players by inducing specific conjectures on players, to reach an equilibrium that maximizes a certain objective function $x\mapsto \mathscr{F}(x)$, e.g., the social welfare or a fairness criterion. A conjecture is defined as a function $$\gamma_i^j:\mathcal{X}_i\to\mathcal{X}_j,\;\forall j \in \mathcal{N}\setminus\{i\},$$ parameterized by $\theta_{i,j} \in \Theta_{i,j}\subseteq \mathbb{R}^d,\,d \in \mathbb{N}^\star$. We define $\theta_i := (\theta_{i,j})_{j\in\mathcal{N}\setminus\{i\}}$, $\gamma_i^{-i}(x_i;\theta_i):=(\gamma_i^j(x_i;\theta_{i,j}))_{j\in\mathcal{N}\setminus\{i\}}$ and $\theta:=(\theta_{i,j})_{i,j \in \mathcal{N},\,i\neq j} \in\Theta:=\prod_{i,j\in\mathcal{N},\,i\neq j} \Theta_{i,j}$.

The coordinator can induce a specific conjecture  $\gamma_i^{-i}(\cdot;\theta_i)$ by choosing the parameter $\theta_i$ and absorbing the discrepancy cost $c_i(x,\theta_i) = J_i(x) - J_i(x_i,\gamma_i^{-i}(x_i;\theta_i))$. Then, the decision problem of the coordinator is given by

\begin{equation}\label{eq:SLMF}
    \begin{array}{cl}
        \displaystyle\max_{x\in\mathbb{R}^m,\theta\in \Theta} & \mathscr{F}(x) - \sum_{i\in\mathcal{N}} c_i(x,\theta_i)  \\
         s.t. & x\in \mathrm{NEP}(\theta).
    \end{array}
    \tag{SLMF}
\end{equation}

where $\mathrm{NEP}(\theta)$ is the solution of the equilibrium problem of the followers, where each player $i\in\mathcal{N}$ is solving
\begin{equation}
    \min_{x_i\in \mathbb{R}^{m_i}} J_i(x,x_{-i}) - c_i(x,\theta).
    \label{eq:players_problem-general}
\end{equation}

This setting is known in the literature as Single-Leader-Multi-Follower games, introduced in \citep{Pang-Fukushima2005}. See also \citep{Aussel-Svensson} for a recent state-of-the-art. The advantage of this formulation is that, even though the problem of the $N$-players is an equilibrium problem, it is decoupled in independent optimization problems thanks to the conjectures. That is, the objective function of each player $i\in\mathcal{N}$ becomes independent of the other players, since by construction
\[
J_i(x,x_{-i}) - c_i(x,\theta_i) = J_i(x_i,\gamma_i^{-i}(x_i;\theta_i)).
\]

In this work, the discrepancy costs between the conjectures and the actual best-responses are handled through the notion of consistency. 
The definitions of consistency are introduced from the weakest to the strongest. We denote with $x^*$ the solution chosen by the coordinator. The weakest level of consistency is the following:
\begin{definition}[Zero-th order consistency]
Given a strategy profile $x^*$, a vector of conjectures $(\gamma_i^j(\cdot,\theta_{i,j}))_{i,\neq j}$ is zero-th order consistent at $x^*$ if
\begin{equation}
    J_i(x_i^*,\gamma_i^{-i}(x_i^*;\theta_i)) = J_i(x_i^*, x_{-i}^*),\;\forall i\in\mathcal{N}.
    \label{eq:0_consistency}
\end{equation}
That is, $c_i(x^*,\theta_i) = 0$ for every $i\in\mathcal{N}$.
\label{def:0_consistency}
\end{definition}
In Definition~\ref{def:0_consistency}, we require that the objective function of each player in $\mathcal{N}$, evaluated at equilibrium through the conjectures, coincides with the real value at $x^*$. This definition implies that a player observing the conjectured cost of its actions and the real one, would not find any discrepancies, which is fundamental to have players not changing their strategy at equilibrium.

A stronger consistency definition is the following:
\begin{definition}[First order consistency]
Given a strategy profile $x^*$, a vector of conjectures $(\gamma_i^j(\cdot,\theta_{i,j}))_{i,\neq j}$ is first order consistent at $x^*$ if
\begin{equation}
    \gamma_i^j(x_i^*;\theta_{i,j}) = x_j^*,\;\forall i,j\in\mathcal{N},i\neq j.
    \label{eq:1_consistency}
\end{equation}
\label{def:1_consistency}
\end{definition}
We notice that Definition~\ref{def:1_consistency} implies Definition~\ref{def:0_consistency}. With this level of consistency each player would still be satisfied with its choice at equilibrium even if it was shown the actual decisions of the other players. The final formulation allows us to remove the discrepancy costs of the objective function of the coordinator, by imposing consistency constraints, that is, aiming to solve

\begin{equation}\label{eq:g_SLMF}
    \begin{array}{cl}
        \displaystyle\max_{x\in\mathbb{R}^m,\theta\in \Theta} & \mathscr{F}(x)  \\
         s.t. & \begin{cases}
         x\in \mathrm{NEP}(\theta),\\
          \gamma_i^j(x_i;\theta_{i,j}) = x_j,\;\forall i,j\in\mathcal{N},i\neq j.
         \end{cases}
    \end{array}
    \tag{$\gamma$-SLMF}
\end{equation}

\section{Centralized Conjecture Design}
We formalize the conjecture design problem using Definition~\ref{def:1_consistency}. Thus, weaker definitions of consistency are immediately checked. 

We start by stating some assumptions:
\begin{assumption}\label{assump_1}
    We assume $\mathscr{F}(x)$ is lower semi-continuous on $\mathcal{X}$ and coercive, $J_i(x_i,x_{-i})$ are $\mathcal{C}^1$, for any $x_{-i}$, the map $-J_i(\cdot,x_{-i})$ is coercive, $\gamma_i^j(x_i;\theta_{i,j})$ is continuous and $\gamma_i^j(\cdot,\theta_{i,j})\in\mathcal{C}^1,\forall i,j\in\mathcal{N},i\neq j,\forall\theta_{i,j}\in\Theta_{i,j}$. Further we assume $\Theta_{i,j}$ is compact for any $i,j\in\mathcal{N}$.
    \label{ass:basic}
\end{assumption}
Under Assumption \ref{assump_1}, the considered Single-Leader-Multi-Follower game, where the coordinator is the leader and the $N$ players are the followers, can be relaxed into the following MPCC format (Mathematical Programming with Complementarity Constraints), see e.g. \citep{Aussel-Svensson}:
\begin{subequations}
    \begin{alignat}{2}
        &\max_{\theta, x} && \quad \mathscr{F}(x),\label{eq:objective}\\
        & \mbox{s.t. } && \quad \nabla J_i(x_i,\gamma_i^{-i}(x_i;\theta_i)) = 0,\;\forall i\in\mathcal{N},\label{eq:stationarity_cond}\\
        & && \quad \gamma_i^j(x_i;\theta_{i,j}) = x_j,\;\forall i,j\in\mathcal{N},i\neq j,\label{eq:1_consistency_constr}
    \end{alignat}
    \label{eq:conj_opt}
\end{subequations}
where Eq.~\eqref{eq:stationarity_cond} takes the explicit form
\begin{equation}
    \begin{split}
        &\nabla J_i(x_i,\gamma_i^{-i}(x_i;\theta_i)) \\
        =& \nabla_i J_i(x_i,\gamma_i^{-i}(x_i;\theta_i))+\nabla_{-i}J_i(x_i,\gamma_i^{-i}(x_i;\theta_i))^\top\nabla_i \gamma_i^{-i}(x_i;\theta_i) = 0,\forall i\in\mathcal{N}.
    \end{split}
    \label{eq:equilibrium_consistency}
\end{equation}
Problem~\eqref{eq:conj_opt} becomes an equivalent reformulation of problem~\eqref{eq:g_SLMF} under general convexity assumptions on the functions $x_i\mapsto J_i(x_i,\gamma_i^{-i}(x_i,\theta_i))$, as we discuss in Section~\ref{subsec:inducing_eq}. Nevertheless, even if both formulations fail to be equivalent, Problem~\eqref{eq:conj_opt} deserves to be study as a natural relaxation of~\eqref{eq:g_SLMF}.

We focus first on obtaining existence results for problem~\eqref{eq:conj_opt}, showing then how the coordinator can induce the solution of \eqref{eq:conj_opt} to be an equilibrium for the players.
\begin{remark}
\label{remark_open_sets}
    Constraint~\eqref{eq:stationarity_cond} holds because we have set $\mathcal{X}_i=\mathbb{R}^{m_i}$, but it could be generalized to the case where $\mathcal{X}_i$ is an open set, with $-J_i(\cdot,x_{-i})$ still being coercive on it. We will use this generalization in some of the numerical examples that will be subsequently introduced in Subsection \ref{subsec:coord_game}.
\end{remark}

\subsection{Proof of Existence}
In this section, we tackle the problem of proving the existence of a solution $(x^*,\theta^*)$ for problem~\eqref{eq:conj_opt}. 
We define the feasible set $\Omega := \{x\in\mathcal{X},\theta\in\Theta:\mbox{Eq.~\eqref{eq:stationarity_cond} and \eqref{eq:1_consistency_constr} hold}\}$. 

According to Assumption \ref{assump_1}, the feasible set of problem \eqref{eq:conj_opt} is closed. Let us now prove that it is nonempty. In order to do so, we start by stating some assumptions on the class of functions for the conjectures:
\begin{assumption}
    For any $x_i\in\mathcal{X}_i,x_j\in\mathcal{X}_j,\alpha_i^j\in\mathbb{R}^{d}$, there exists $\theta_{i,j}\in\Theta_{i,j}$ such that:
    \begin{equation}
        \begin{cases}
            & \nabla_i\gamma_i^j(x_i;\theta_{i,j}) = \alpha_i^j,\\
            & \gamma_i^j(x_i;\theta_{i,j}) = x_j.
        \end{cases}
        \label{eq:conj_conditions}
    \end{equation}
    \label{ass:conjecture_class}
\end{assumption}
We present examples of conjecture classes that satisfy Assumption~\ref{ass:conjecture_class}, demonstrating that this requirement is not too restrictive. In particular, the two cases we discuss correspond to the classical conjecture types commonly used in conjectural games.
\begin{example}[Affine Conjectures]
    Consider the class of conjectures $\gamma_i^j(x_i;\theta_{i,j})=a_{i,j}+b_{i,j}x_i$, then the conditions in Eq.~\eqref{eq:conj_conditions} are satisfied by:
    \begin{equation*}
    \theta_{i,j} = 
        \begin{bmatrix}
            a_{i,j}\\
            b_{i,j}
        \end{bmatrix}
        =
        \begin{bmatrix}
            x_j - \alpha_i^jx_i\\
            \alpha_i^j
        \end{bmatrix}
    \end{equation*}
\end{example}
\begin{example}[Quadratic Conjectures]
    Another example is the class $\gamma_i^j(x_i;\theta_{i,j})=a_{i,j}+b_{i,j}x_i^2$, which results in the parameter choice:
    \begin{equation*}
    \theta_{i,j} = 
        \begin{bmatrix}
            a_{i,j}\\
            b_{i,j}
        \end{bmatrix}
        =
        \begin{bmatrix}
            x_j - (\alpha_i^jx_i)/2\\
            \alpha_i^j/(2x_i)
        \end{bmatrix}
    \end{equation*}
\end{example}
Before stating the main result of this section, we need a final assumption on the objective functions of the players:
\begin{assumption}
    For every player $i\in\mathcal{N}$, there exists at least one point $\hat{x}^i$ where $J_i(\hat{x}^i_i,\hat{x}^i_{-i})=0$ and the partial Jacobian $\nabla_{-i}J_i(\hat{x}^i_i,\hat{x}^i_{-i})$ is full rank, or, equivalently in our formulation, is non-zero. 
    \label{ass:implicit_func}
\end{assumption}
Notice that requiring the equality to 0 is a mild assumption, as it can be any constant value $c$, and we can rescale the function $J_i'(\hat{x}^i_i,\hat{x}^i_{-i})=J_i(\hat{x}^i_i,\hat{x}^i_{-i})-c=0$, since we are interested in working with the derivative of $J_i$, and $\nabla J_i'=\nabla J_i$. Assumption~\ref{ass:implicit_func} is very classical, and it does not restrict the class of problems we can tackle with this framework. We can now state the following result:
\begin{proposition}
    Suppose Assumptions~\ref{ass:conjecture_class} and \ref{ass:implicit_func} hold. Then, the feasible set $\Omega$ is nonempty.
    \label{prop:non_empty}
\end{proposition}
\begin{proof}
    We evaluate the constraints \eqref{eq:stationarity_cond} and \eqref{eq:1_consistency_constr} in $\hat{x}^i$, the point given by Assumption~\ref{ass:implicit_func}. If we write \eqref{eq:stationarity_cond} using the explicit response functions of the other players (from the perspective of player $i$), we obtain:
    \begin{equation}
    \begin{split}
            &\nabla J_i(\hat{x}^i_i,x_{-i}(\hat{x}^i_i)) = \nabla_iJ_i(\hat{x}^i_i,x_{-i}(\hat{x}^i_i)) + \nabla_{-i}J_i(\hat{x}^i_i,x_{-i}(\hat{x}^i_i))\nabla_i x_{-i}(\hat{x}^i_i) = 0.
    \end{split}
        \label{eq:real_derivative}
    \end{equation}
    Comparing it with Eq.~\eqref{eq:equilibrium_consistency} and using the consistency constraint in Eq.~\eqref{eq:1_consistency_constr}, we can observe that in order to satisfy this constraint we need only to compare the term $\nabla_i x_{-i}(\hat{x}^i_i)$ to $\nabla_i\gamma_i^{-i}(\hat{x}^i_i;\theta_i)$. If we assume the existence of $\nabla_i x_{-i}(\hat{x}^i_i)$, then we end up in the case covered by Assumption~\ref{ass:conjecture_class}, since both constraints can be framed as in Eq.~\eqref{eq:conj_conditions}, where $\alpha_i^{j}=\nabla_i x_{j}(\hat{x}^i_i), x_{j}=\hat{x}_{j}^i,\;\forall j\in\mathcal{N}, j\neq i$.
    
    Regarding the existence of the term $\nabla_i x_{j}(\hat{x}^i_i)$, we can use the generalized implicit function theorem \citep[Th. 1F.6, Ex. 1F.9]{doncev_implicit_2009}, since in $\hat{x}^i$ Assumption~\ref{ass:implicit_func} is satisfied by construction. Considering $J_i(\hat{x}^i_i,\hat{x}^i_{-i})=0$, we can define the solution mapping:
    \begin{equation*}
        S_i: x_i \to \{x_{-i}\in\mathcal{X}_{-i}|J_i(x_i,x_{-i})=0\}, x_i\in\mathcal{X}_i,
    \end{equation*}
    obtaining $\hat{x}_{-i}^i\in S_i(\hat{x}_i^i)$. Then $S_i$ has a local selection $s_i$ around $\hat{x}_i^i$ for $\hat{x}_{-i}^i$, strictly differentiable in $\hat{x}_i^i$, with Jacobian:
    \begin{equation*}
        \nabla s_i(\hat{x}_i^i) = g(\hat{x}^i)^{\top} \left(g(\hat{x}^i) g(\hat{x}^i)^{\top} \right)^{-1} \nabla_i J_i(\hat{x}^i_i,\hat{x}^i_{-i}),
    \end{equation*}
    where $g(\hat{x}^i)=\nabla_{-i}J_i(\hat{x}^i_i,\hat{x}^i_{-i})$. The local selection $s_i(x_i)$ is exactly what we previously called $x_{-i}(x_i)$, which concludes the proof.
\end{proof}

By the Weierstrass extreme value theorem and under the previous propositions and assumptions, we obtain the following existence result:
\begin{theorem}
    Suppose Assumptions~\ref{ass:basic}--\ref{ass:implicit_func} hold. Then, the conjecture design problem~\eqref{eq:conj_opt} admits at least a  solution.
    \label{th:existence}
\end{theorem}
This framework allows us to obtain an existence guarantee with very few steps and minimal assumptions, which is in stark contrast to the framework discussed in \citep{figuieres_theory_2004}. For instance, in \citep{calderone_consistent_2023}, the existence of a consistent equilibrium for a quadratic game with affine conjectures, is linked to the solution of coupled Riccati equations, so even for specific and easier settings (quadratic game), existence is not trivially obtained. In addition, until now we only have Assumptions~\ref{ass:basic}~and~\ref{ass:implicit_func} limiting the choices for each $J_i$, and both assumptions are very mild.

\subsection{Inducing Equilibrium}
\label{subsec:inducing_eq}
Once the solution of problem~\eqref{eq:conj_opt}, $(x^*,\theta^*)$, is obtained, the conjectures parametrized by $\theta^*$ are given to the players, and they solve the problem:
\begin{equation}
    \forall i\in\mathcal{N}: \min_{x_i}J_i(x_i,\gamma_i^{-i}(x_i;\theta_i^*))
    \label{eq:players_problem}
\end{equation}
obtaining the Nash equilibrium $x^N(\theta^*)=\{x_i^N(\theta_i^*)\}_{i\in\mathcal{N}}$. We need now to prove that $x^N(\theta^*)=x^*$. Notice that the choice of constraints to obtain $(x^*,\theta^*)$ does not influence how the players use the conjectures and play the game, hence the following results hold regardless.
For a given $\theta^*$, since problem \eqref{eq:players_problem} is unconstrained, its KKT conditions take the form
\begin{equation}
    \begin{split}
        \forall i\in\mathcal{N}&: \nabla_i J_i(x_i,\gamma_i^{-i}(x_i;\theta_i^*)) +\nabla_{-i}J_i(x_i,\gamma_i^{-i}(x_i;\theta_i^*))^\top\nabla_i\gamma_i^{-i}(x_i;\theta_i^*) = 0,
    \end{split}
    \label{eq:derivative_conj_obj}
\end{equation}
which just consists in the full derivatives of the conjectured objective functions $J_i(x_i,\gamma_i^{-i}(x_i;\theta_i^*))$.

Without any assumption on $J_i(x_i,\gamma_i^{-i}(x_i;\theta_i^*))$, we cannot conclude that a solution of Eq.~\eqref{eq:derivative_conj_obj} is an equilibrium for the optimization problem $\eqref{eq:players_problem}$. We first recall the definition of a pseudo-convex function:
\begin{definition}[Pseudo-Convexity \citep{mangasarian_pseudo-convex_1965}]
    Let $f:\mathcal{X} \subseteq \mathbb{R}^m \rightarrow \mathbb{R}$ be a differentiable function on $\mathcal{X}$. This function is said to be pseudo-convex if
    \begin{equation*}
        \nabla f(a)^\top(b-a)\geq 0 \Rightarrow f(b) \geq f(a), \quad \forall a, b \in \mathcal{X}.
    \end{equation*}
\end{definition}
We can then state our assumption:
\begin{assumption}
    For each player $i\in\mathcal{N}$, given $\theta^* \in \Theta$, the conjectured objective function $J_i(x_i,\gamma_i^{-i}(x_i;\theta_i^*))$ is pseudo-convex in $x_i$.
    \label{ass:pseudo_convex}
\end{assumption}
Notice that the previous assumption only refers to the conjectured functions, meaning that the game $\mathcal{G}$ need not be pseudo-convex. We can show an example where only the conjectured objective function is pseudo-convex, but not the original one:
\begin{example}[Saddle Function]
    Consider $J(x,y)=-xy$, which gives the gradient $\nabla J(x,y)=(-y,-x)$. If we take the points $a=(0,0),b=(1,1)$ and compute the terms in the definition of pseudo-convexity, we obtain:
    \begin{equation*}
        \begin{split}
            \nabla J(a)(b-a) = 0,\\
            J(b)=-1< 0 = J(a).
        \end{split}
    \end{equation*}
    so $J(x,y)$ is not pseudo-convex. If now we consider the conjecture $\gamma_x^y(x) = -x$, we obtain $J(x, \gamma_x^y(x))=-x\gamma_x^y(x)=x^2$, which is convex, so by definition, also pseudo-convex.
\end{example}
\begin{remark}
Assumption~\ref{ass:pseudo_convex} is stated directly on the composed conjectured objective $J_i(x_i,\gamma_i^{-i}(x_i))$, since pseudo-convexity generally depends on the interaction between $J_i$ and $\gamma_i$ rather than on each function separately. Broad characterizations in terms of $J_i$ and $\gamma_i$ alone are therefore difficult to obtain. Nevertheless, standard sufficient conditions are available in the literature (see, e.g., \mbox{\citep{cambini_pseudoconvexity_2002}}).
\end{remark}
We can now state our second key result:
\begin{theorem}
    Given Assumption~\ref{ass:pseudo_convex}, the equilibrium $x^N(\theta^*)$ reached by playing the game described in Eq.~\eqref{eq:players_problem} is equivalent to a solution $x^*$ obtained by coordinator solving Eq.~\eqref{eq:conj_opt}.
    \label{th:equivalence}
\end{theorem}
\begin{proof}
    A solution $(x^*,\theta^*)$ of the problem in \eqref{eq:conj_opt} solves, by construction, $\nabla J_i(x_i^*,\gamma_i^{-i}(x_i^*;\theta_i^*))=0$ for each $i\in\mathcal{N}$, which implies that it solves the stationarity equation (Eq.~\eqref{eq:derivative_conj_obj}) associated to problem~\eqref{eq:players_problem}. Then from Assumption~\ref{ass:pseudo_convex}, we guarantee that a solution of Eq.~\eqref{eq:derivative_conj_obj} is a global solution for the optimization of each player $i$.
\end{proof}
The result stated in Theorem~\ref{th:equivalence} is very important for the applicability of our framework: we can now guarantee that if the players solve the game with the conjectures designed by the coordinator, they will converge towards the desired outcome.
In Game~\ref{game:centralized} we describe the timing of the centralized conjecture design problem.

\begin{listing}[tb]%
\caption{\textit{Centralized Conjecture Design}}%
\label{game:centralized}%
\begin{lstlisting}[mathescape, language=, basicstyle=\rmfamily, columns=flexible]
Step 1: Coordinator solves problem (3), obtaining $\theta^*$
Step 2: Each player $i$ optimizes $J_i(x_i,\gamma_i^{-i}(x_i;\theta_i^*))$ independently
Step 3: Obtain Nash equilibrium $x^N(\theta^*)$
\end{lstlisting}
\end{listing}

\subsection{Strengthening Conjecture Consistency}
We propose a third consistency definition, closer to the original version in \citep{figuieres_theory_2004}, which considers directly the best-response function of each player. For player $i$, we refer to its best-response with $x_i(x_{-i})$, overloading the notation for the strategy and expliciting the dependence on the other players strategies. We assume the following:
\begin{assumption}
     The best-response function $x_i(x_{-i})$ exists, is unique, continuous and differentiable for every player $i\in\mathcal{N}$.
    \label{ass:best_response}
\end{assumption}
Assumption~\ref{ass:best_response} is very strong and limiting, which is why we only consider it as an extension of the main model. We can now state the last consistency definition:
\begin{definition}[Second order consistency]
    Given a strategy profile $x^*$, a vector of conjectures $(\gamma_i^j)_{i,\neq j}$ is second order consistent at $x^*$ if
    \begin{equation}
        \nabla_i\gamma_i^j(x_i^*;\theta_{i,j}) = \nabla_ix_j^*(x),\;\forall i,j\in\mathcal{N},i\neq j,
        \label{eq:2_consistency}
    \end{equation}
    \label{def:2_consistency}    
\end{definition}
Definition~\ref{def:2_consistency} does not imply the previous ones, so if used, it must be added on top of either Definition~\ref{def:0_consistency} or \ref{def:1_consistency}. Notice that, if we add Eq.~\eqref{eq:2_consistency} to problem \eqref{eq:conj_opt}, we can use Assumption~\ref{ass:best_response} to directly extend \ref{prop:non_empty}. In particular, we no longer need Assumption~\ref{ass:implicit_func} to prove the feasible set is non-empty, since we already assume everything needed on $x_{-i}(x)$. We define the new feasible set $$\bar{\Omega}= \{x\in\mathcal{X},\theta\in\Theta:\mbox{Eq.~\eqref{eq:stationarity_cond}, \eqref{eq:1_consistency_constr} and \eqref{eq:2_consistency} hold}\},$$ and we can state:
\begin{proposition}
    Under Assumptions~\ref{ass:basic}, \ref{ass:conjecture_class} and \ref{ass:best_response}, the conjecture design problem described in \eqref{eq:conj_opt} with the addition of constraint~\eqref{eq:2_consistency} admits at least a solution.    
\end{proposition}

\section{Decentralized Conjecture Design}
Framing the conjecture design as a centralized problem makes it easier to study and characterize the equilibria, but in practical application, players may not trust, or accept, conjectures given by a central entity. Instead, decentralizing the conjectures design, could lead to a more realistic model. The coordinator cannot be completely removed, since to compute the parameters of the conjectures, the players must know what is the new equilibrium $x^*$ to reach. Specifically, the only information that player $i$ needs is the value $J_i(x_i^*,x_{-i}^*)$, which is used as a target. We stress that, from the perspective of player $i$, $J_i(x_i^*,x_{-i}^*)$ is just a target value, so it does not have access to neither $x_i^*$ nor $x_{-i}^*$. We can frame the problem for each player $i$ as:
\begin{equation}
    \begin{cases}
        &\nabla J_i(x_i,\gamma_i^{-i}(x_i;\theta_i)) = 0,\\
        & J_i(x_i,\gamma_i^{-i}(x_i;\theta_i)) - J_i(x_i^*,x_{-i}^*) = 0,
    \end{cases}
    \label{eq:decentralized_problem}
\end{equation}
where player $i$ solves the set of equations for $(x_i,\theta_i)$. If we compare the system of equations \eqref{eq:decentralized_problem} with Eqs.~\eqref{eq:conj_opt}, we can see that each player takes care of the constraints, while the coordinator is left with only the optimization of $\mathscr{F}(x)$ to find $x^*$. We refer to the solution of Eqs.~\eqref{eq:decentralized_problem} as $(\tilde{x}_i,\tilde{\theta}_i),\forall i\in\mathcal{N}$. Since we are dealing with the same constraints as before, we obtain the existence guarantee for free, just considering the same assumptions:
\begin{proposition}
    Given Assumptions~\ref{ass:basic}--\ref{ass:implicit_func}, problem \eqref{eq:decentralized_problem} admits a solution $(\tilde{x}_i,\tilde{\theta}_i),\forall i\in\mathcal{N}$.
    \label{prop:existence_decentral}
\end{proposition}
We report in Game~\ref{game:decentralized} the timing of the decentralized conjecture design problem.

\begin{listing}[tb]%
\caption{\textit{Decentralized Conjecture Design}}%
\label{game:decentralized}%
\begin{lstlisting}[mathescape, language=, basicstyle=\rmfamily, columns=flexible]
Step 1: Coordinator optimizes $\mathscr{F}(x)$, obtaining $x^*$
Step 2: Coordinator passes to each player $i$ the target $J_i(x^*)$
Step 2: Each player $i$ solves Eqs.(10) independently
Step 3: Obtain solution $(\tilde{x},\tilde{\theta})$
\end{lstlisting}
\end{listing}

\subsection*{Discussion: Full Decentralization}
In order to solve problem~\eqref{eq:conj_opt}, the coordinator has access to the derivatives of the objective functions of the players and their actions, while it does not need the actual shape of the singular objective functions themselves, since it optimizes $\mathscr{F}(x)$ directly. In problem~\eqref{eq:decentralized_problem}, we reduce the amount of information the coordinator has access by limiting its knowledge to the objective functions, removing the derivatives. This problem can then be seen as semi-decentralized, since there is still some information that needs to be known by the coordinator. Nonetheless, both the centralized case and the (semi-)decentralized one generalize classical incentive design problems, such as tolling \citep{roughgarden_how_2002}, power grids optimization \citep{alshehri_stackelberg_2015}, or competition between firms \citep{ratliff_adaptive_2021}.
We can still discuss an ideal, fully decentralized model, where the coordinator has access only to the predefined objective function $\mathscr{F}(x)$, and it sends a target signal to each player, that we may call $\bar{J}_i$. This signal is not a function, it is a value decided by the coordinator that replaces the real $J_i(x)$ in Eq.~\eqref{eq:decentralized_problem}. Framing the model this way, does not require the coordinator to have any knowledge about the players' problem, but it creates a new issue: how can the coordinator choose the values $(\bar{J}_i)_i$? Indeed, assuming the players compute their conjectures with parameters $\tilde{\theta}$ and their solution strategy $\tilde{x}$, the coordinator can compute how far they are from the intended equilibrium $x^*$, evaluating the difference $\Delta = \mathscr{F}(x^*)-\mathscr{F}(\tilde{x})$. However, the value of $\Delta$ gives no information about how far from $x^*$ \textit{each player} is. 
In symmetric games, $\Delta$-based update rules may work. However, in general, more sophisticated feedback mechanisms need to be defined to allow full conjecture design decentralization. 

\section{Applications}
The game instances we select showcase situations where the players cannot achieve social optimum without external incentives or coordination. On top of that, relying on conjectural games, each player only has to solve its own optimization problem. Each game instance can be solved equivalently relying on either the centralized formulation from Eqs.~\eqref{eq:conj_opt}, or with the decentralized formulation from Eq.~\eqref{eq:decentralized_problem}.
While the two first instances can be solved analytically; the last instance can only be tackled numerically. We leave in the code repository\footnote{\url{https://github.com/FrancescoMorri/ConjecturesDesign}} all the scripts needed to reproduce the results and the graphs.

\subsection{Tragedy of the Commons}
Consider the classic example of the tragedy of the commons \citep{hardin_tragedy_1968}. Multiple actors share a common good, of which they can take a part: an actor reward increases with its share of the common resource, but at the same time taking too much of the resource results in a worse outcome for everyone. The original model is supposed to represent industries polluting the environment, where higher production rates increase their revenues while also increasing global pollution, which in turns makes everyone worst off.
Considering two maximizing players, we can write their objective functions as
\begin{equation*}
    J_i(x) = \ln(x_i) + \ln(K-x_i-x_j),
\end{equation*}
where $K$ is a parameter representing the quantity of a common good, while $x_i$ is the share of common good belonging to player $i$. The \textit{tragedy} is encoded in the difference between the Nash equilibrium (NE) and the social optimum (SO): the NE is $x^N_1=x_2^N=K/3$, while the SO is $x_1^S=x_2^S=K/4$. The SO uses less of the common good, nonetheless $J_i(x^S)>J_i(x^N)$ for both players. We can tackle and solve this problem with our conjectures design optimization. Notice that we can apply our framework because the feasible sets for both players are open sets and the objective functions are coercive on them, as explained in Remark~\ref{remark_open_sets}. We consider linear conjectures $\gamma_i^j(x_i)=a_i+b_ix_i$ and tune the parameters $a_i,b_i$ so that Eqs.~\eqref{eq:stationarity_cond}-\eqref{eq:1_consistency_constr} hold in $(x_1,x_2)=(K/4,K/4)$, obtaining $a_1=a_2=0$ and $b_1=b_2=1$. The conjectured objective function becomes
\begin{equation*}
    J_i(x_i,\gamma_i^j(x_i))=\ln(x_i)+\ln(K-2x_i),
\end{equation*}
and we can observe that $\gamma_i^j(K/4)=K/4;\;i,j=1,2$ as constrained by Eq.~\eqref{eq:1_consistency_constr}, and $\nabla J_i(K/4,\gamma_i^j(K/4))=0$ as constrained by Eq.~\eqref{eq:stationarity_cond}. We can then use Theorem~\ref{th:equivalence} to conclude that two players solving the conjectured game with $J_i(x_i,\gamma_i^j(x_i))$ would converge to the solution $(K/4,K/4)$.

\begin{figure*}[bth]
\begin{subfigure}[t]{.32\linewidth}
    \includegraphics[width=\linewidth]{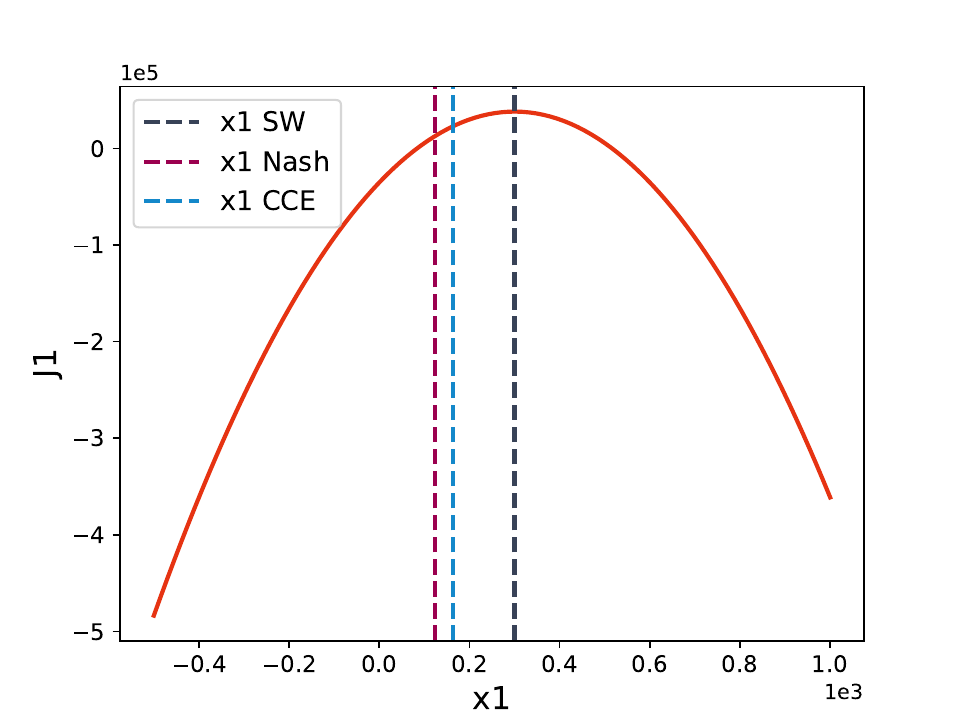}
\caption{Plot of $J_1(x_1,\gamma_1^2(x_1))$ for the Olsder's paradox.}
\label{fig:olsder_conj_j1}
\end{subfigure}
\hfill
\begin{subfigure}[t]{.32\linewidth}
    \includegraphics[width=\linewidth]{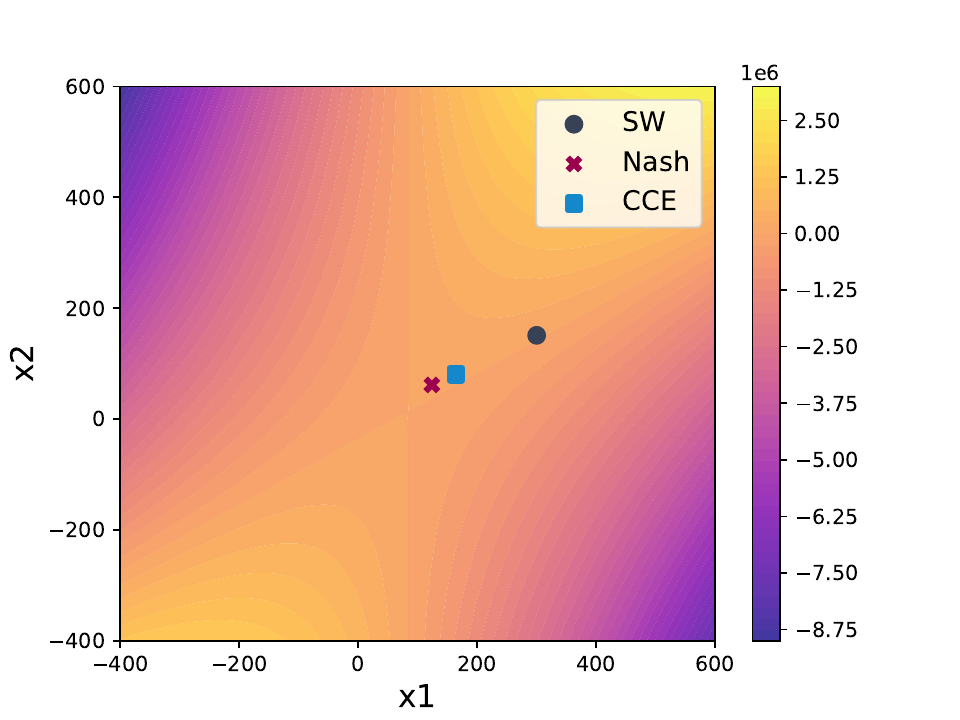}
\caption{Contour plot of the \textit{real} $J_1(x_1,x_2)$.}
\label{fig:olsder_j1}
\end{subfigure}
\hfill
\begin{subfigure}[t]{.32\linewidth}
    \includegraphics[width=\linewidth]{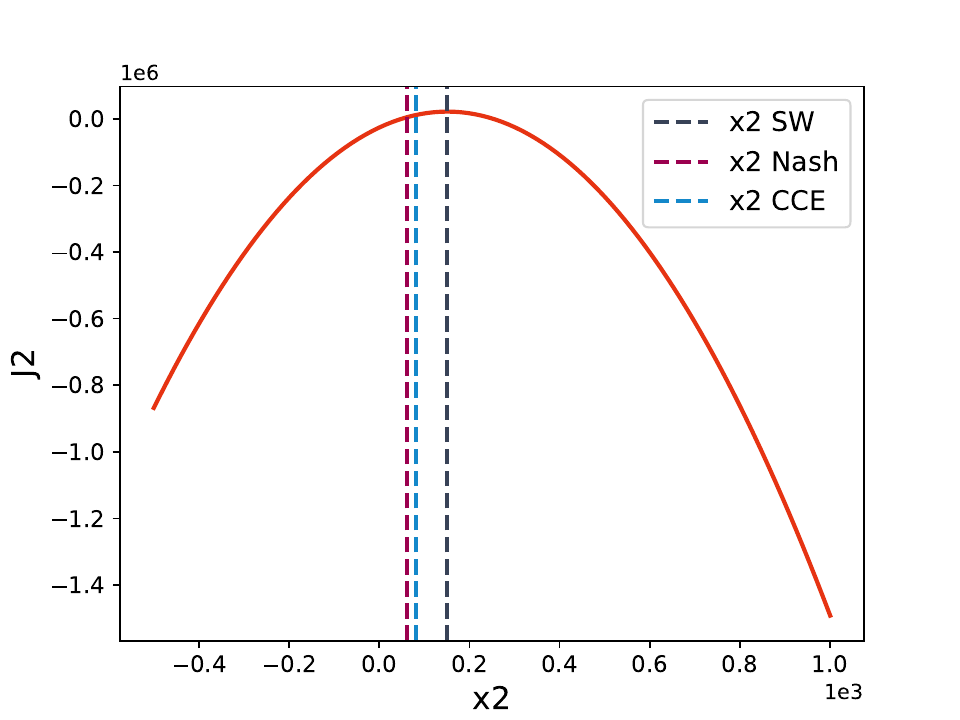}
    \caption{Plot of $J_2(\gamma_2^1(x_2), x_2)$ for the Olsder's paradox.}
\label{fig:olsder_conj_j2}
\end{subfigure}
\caption{We report the conjectured objective functions in Fig.~\ref{fig:olsder_conj_j1}-\ref{fig:olsder_conj_j2}, and the real one for player 1 in Fig.~\ref{fig:olsder_j1}. In all plots we highlight the position of the equilibria.}
\label{fig:olsder}
\end{figure*}

\subsection{Olsder's Paradox}
We consider an intriguing example proposed in \citep{olsder_about_2010}, where two maximizing players, with $x_1,x_2\in\mathcal{X}\subseteq\mathbb{R}$, have the objective functions:
\begin{equation}
    \begin{split}
        &J_1(x_1,x_2)=(x_1-84)(-12.5x_1+21x_2+756)\\
        &J_2(x_1,x_2)=(x_2-50)(24x_1-50x_2+560).
    \end{split}
\end{equation}
This game is interesting because it was built specifically to show how a consistent conjectural equilibrium (in the sense of \citep{figuieres_theory_2004}) produced better payoff that the Nash equilibrium. Specifically, the NE of the game is $(x_1^N,x_2^N)=(123.98,61.6)$, which gives returns $J_1 \simeq  19984, J_2 \simeq 5284$, while the consistent conjectural equilibrium (CCE) is $(x_1^C, x_2^C)=(164.4,81)$, with returns $J_1 \simeq  32321, J_2 \simeq 14124$. The main issue with the CCE is that, in order to obtain it, a rather difficult set of second order equations must be solved, considering the derivative of the best response of each player, which is not a function usually accessible during the game. In contrast, if we apply our framework, we can easily obtain consistent conjectures (in the sense of Eq.~\eqref{eq:1_consistency}), that create an equilibrium in the social welfare optimum (SO) $(x_1^S,x_2^S)=(300.04, 150.98)$, computed maximizing $J_1+J_2$. Solving Eqs.~\eqref{eq:stationarity_cond}-\eqref{eq:1_consistency_constr}, we obtain the linear conjectures:
\begin{equation*}
    \begin{split}
        &\gamma_1^2(x_1) = -15.9704 + 0.5564x_1\\
        &\gamma_2^1(x_2) = -1.2970 + 1.9959x_2,
    \end{split}
\end{equation*}
and we can easily check that $\gamma_1^2(300.04)=150.98$, $\gamma_2^1(150.98)=300.04$. The return obtained by the players using the conjectures is $J_1\simeq38040, J_2\simeq21404$, which is higher than both the NE and CCE. We report a summary of the results in Tab.~\ref{tab:olsder_results} and graphs representing the different equilibria reached (Nash equilibrium, CCE), social welfare optimum, and conjectured objective functions in Fig.~\ref{fig:olsder}. 
\begin{table}[ht]
  \centering
  \begin{tabular}{lcc}
    \toprule
    Equilibrium & $J_1$ & $J_2$ \\
    \midrule
    NE    & 19984     & 5284    \\
    CCE   & 32321     & 14124   \\
    SO   & $\boldsymbol{38040}$     & $\boldsymbol{21404}$    \\
    \bottomrule
  \end{tabular}
    \caption{Returns obtained by both players in NE, CCE and SO.}
    \label{tab:olsder_results}
\end{table}

\subsection{Coordination Game}\label{subsec:coord_game}
We now consider a large-scale instance, that we solve numerically. The players' utility functions have a common element and a user's specific cost:
\begin{equation}
    J_i(x) = -a_i\left(\frac{1}{N}\sum_{j\in\mathcal{N}}x_j-\frac{1}{N}\sum_{j\in\mathcal{N}}d_j\right)^2 - b_i(x_i-d_i),
    \label{eq:coord_obj_2}
\end{equation}
where $x_i\geq0,a_i>0,b_i\geq0$. The first term in $J_i$ represents a common target, while the second one is the player's personal target. This game admits a single symmetric Nash equilibrium if $a_i=a,b_i=b,\forall i\in\mathcal{N}$, which is $x_i^N=\frac{1}{N}\sum_{j\in\mathcal{N}}d_j - \frac{Nb}{2a}$ \citep{s_how_2025}. Further, there exists a single social optimum for the symmetric case, which is $x_i^S=\frac{1}{N}\sum_{j\in\mathcal{N}}d_j - \frac{b}{2a}$. We know there are no asymmetric solution because we would have to solve a system of equations of the form:
\begin{equation*}
    \begin{cases}
        x_1+x_2+\dots+x_N=&A_1\\
        \vdots &\vdots\\
        x_1+x_2+\dots+x_N=&A_N
    \end{cases}
\end{equation*}
which has no solution. In this situation our framework has a double goal: in the asymmetric case we can induce an equilibrium and in the symmetric case we can improve over the Nash. Moreover, we can compare the centralized and decentralized approach. Notice that since there is no NE, in the asymmetric case we lose the guarantee of existence of a solution; while in the symmetric case we refer once again to Remark~\ref{remark_open_sets}, assuming that the parameters characterizing the solution, $(\frac{1}{N}\sum_{j\in\mathcal{N}}d_j, a, b)$, force it to be larger than 0, effectively removing the constraint $x_i\geq0,\forall i\in\mathcal{N}$.
In Tab.~\ref{tab:coord_results} we report the results for both cases evaluated with different amount of players. Each entry in the table refers to the algorithm performing the best: in the symmetric case we also compare against the NE, while in the asymmetric one just between our solutions. For a fixed number of players, the algorithms are compared over the same parameter instance $\{(a_i)_{i\in\mathcal{N}}, (b_i)_{i\in\mathcal{N}}, (d_i)_{i\in\mathcal{N}}\}$. We refer to the algorithm solving problem~\eqref{eq:conj_opt} as \texttt{CS}, to the one replacing constraint~\eqref{eq:1_consistency_constr} with Eq.~\eqref{eq:0_consistency} as \texttt{CW}, and to the decentralized one with \texttt{D}.
We can immediately see that for any $N$ in the symmetric case, we find a solution that is better than the NE. The centralized algorithms are better for lower number of players, while the decentralized one produces the best solutions for the larger instances. Moreover, in many of the asymmetric cases, the centralized approach could not produce an optimal solution, and in some of the largest instances, it could not find a solution at all, while the decentralized always worked.
In conclusion, the conjectural design algorithm yields an equilibrium even when no Nash equilibrium exists and improves equilibrium outcomes when one is available.
\begin{table*}[bth]
  \centering
  \begin{tabular}{lccccccc}
    \toprule
    Settings & $N=2$ & $N=5$ & $N=10$ & $N=15$ & $N=20$ & $N=30$ & $N=50$\\
    \midrule
    Symmetric & \texttt{CW} & \texttt{CW} & \texttt{CS} & \texttt{CS} & \texttt{D} & \texttt{D} & \texttt{D} \\
    Asymmetric & \texttt{CS} & \texttt{CS} & \texttt{CS} & \texttt{D} & \texttt{D} & \texttt{D} & \texttt{D}  \\
    \bottomrule
  \end{tabular}
  \caption{Best algorithm from the Coordination Game.}
  \label{tab:coord_results}
\end{table*}

\section{Discussion and Limitation: Saddle Games}
We now discuss a class of games that requires the stronger consistency proposed in Def.~\ref{def:2_consistency} in order to be solved with our framework. We consider zero-sum saddle games, which we can write as:
\begin{equation}
    J_1(x_1,x_2)=-J_2(x_1,x_2)=(x_1-\bar{x}_1)(x_2-\bar{x}_2),
    \label{eq:saddle_game}
\end{equation}
with both players maximizing. Here the only Nash equilibrium is in $(\bar{x}_1,\bar{x}_2)$, which is saddle point for both players, meaning that is generally complicated to reach. We define the coordinator objective function $\mathscr{F}(x)= J_1J_2$, and without considering second order consistency, we obtain the following optimization problem:
\begin{subequations}
    \begin{alignat}{2}
        &\max_x &&-(x_1-\bar{x}_1)^2(x_2-\bar{x}_2)^2\label{eq:saddle_select}\\
        &\mbox{s.t. } &&(\gamma_1^2(x_1)-\bar{x}_2) + \nabla_1 \gamma_1^2(x_1)(x_1-\bar{x}_1)=0\label{eq:saddle_c1}\\
        & && (\gamma_2^1(x_2)-\bar{x}_1) + \nabla_2 \gamma_2^1(x_2)(x_2-\bar{x}_2) = 0\label{eq:saddle_c2}\\
        & && \gamma_1^2(x_1) = x_2\label{eq:saddle_c3}\\
        & && \gamma_2^1(x_2) = x_1\label{eq:saddle_c4}
    \end{alignat}
\end{subequations}
When we evaluate constraints~\eqref{eq:saddle_c1}-\eqref{eq:saddle_c2} in $(\bar{x}_1,\bar{x}_2)$, with conjectures satisfying constraints~\eqref{eq:saddle_c3}-\eqref{eq:saddle_c4}, they are always satisfied.
If we consider affine conjectures $\gamma_1^2(\bar{x}_1) = a_1 + b_1\bar{x}_1$, it means that we can choose any value for $b_1$ and then just fix $a_1=\bar{x}_2-b_1\bar{x}_1$. This can create issues, for instance the conjectures $\gamma_1^2(x_1)=x_1 + \bar{x}_2-\bar{x}_1, \gamma_2^1(x_2)=-x_2+\bar{x}_1+\bar{x}_2$ satisfy all constraints, but produce the conjectured objective functions:
\begin{equation}
    \begin{split}
        &J_1^{\gamma}(x_1) = x_1^2 + (\bar{x}_2-\bar{x}_1)x_1,\\
        &J_2^{\gamma}(x_2) = x_2^2 + (\bar{x}_1+\bar{x}_2)x_2,
    \end{split}
\end{equation}
which in maximization settings are problematic, being convex. Adding the second order consistency constraint solves this problem. For the saddle game~\eqref{eq:saddle_game}, Eq.~\eqref{eq:2_consistency} constrains the slopes of the affine conjectures to be $b_1 = -1, b_2 = 1$, so we obtain the functions: $\gamma_1^2(x_1) = -x_1 + \bar{x}_2 + \bar{x}_1, \gamma_2^1(x_2) = x_2 + \bar{x}_1-\bar{x}_2$. The conjectured objective functions are now:
\begin{equation}
    \begin{split}
        &J_1^{\gamma}(x_1) = -x_1^2 + (\bar{x}_2+\bar{x}_1)x_1,\\
        &J_2^{\gamma}(x_2) = -x_2^2 - (\bar{x}_1-\bar{x}_2)x_2,
    \end{split}
    \label{eq:saddle_conj_objective}
\end{equation}
which are good solutions for the game, since they are concave with a maximum in $(\bar{x}_1,\bar{x}_2)$ for $J_1,J_2$ respectively.
\begin{figure}[h]
    \centering
    \includegraphics[width=0.7\linewidth]{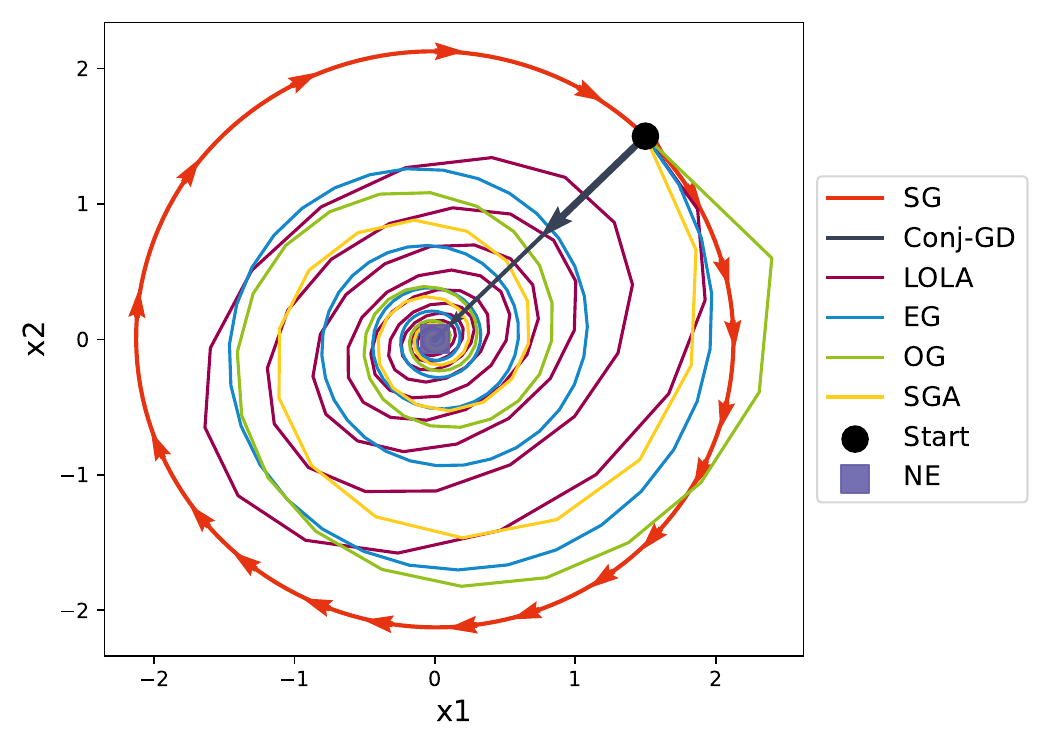}
    \caption{Comparison of learning a strategy in the saddle game for different algorithms. Each one is ran for $10^3$ steps with different learning steps.}
    \label{fig:saddle_compare}
\end{figure}
In Fig.~\ref{fig:saddle_compare} we report a comparison between different algorithms used to learn a strategy in a saddle game with $\bar{x}_1=\bar{x}_2=0$. \textit{Conj-GD} is gradient ascent on the conjectured objective functions in \eqref{eq:saddle_conj_objective}, while \textit{SG} (simultaneous gradient), \textit{LOLA} (learning with opponent-learning awareness), \textit{EG} (extragradient), \textit{OG} (optimistic gradient) and \textit{SGA} (symplectic gradient adjustment) are benchmark algorithms taken from \citep{martin_approxed_2025}, from which we also took the example of the saddle game. For each algorithm we choose hyperparameters that ensured convergence, which we report in the code repository. When good conjectures are available, our framework enables extremely fast convergence to the game's equilibrium, with the added benefit that each player's learning process can be parallelized.

\section{Conclusions}
In this work we propose a novel framework unifying concepts from mechanism, incentive design, and opponent modeling; grounded in conjectural variations games. We derive theoretical guarantees for the existence of conjectures satisfying different levels of consistency constraints, under standard regularity assumptions about the game. This is in stark contrast to previous results on conjectural variations games, where obtaining proof of existence for conjectures was extremely complicated, and needed strong assumptions.

We consider both a centralized and a decentralized approach for designing the conjectures: in the first, a coordinator computes the conjectures' parameters, while in the second the coordinator only communicates a target to the players, who compute their own conjectures. Both formulations share the same existence guarantee, and for the centralized approach we also prove that the equilibrium reached by players using the conjectures is indeed the one designed by the coordinator. Moreover, these formulations produce a set of independent optimization problems for the players, that can be solved in parallel for large scale simulations.
We tested our model on multiple games, that we selected to highlight specific problems: either the NE is worse than the social welfare, or the NE is hard to reach without coordination, or, even, there is no NE at all. In all cases our solutions (centralized and decentralized) improve the outcome for the players.

Finally, we discuss full decentralization as a possible extension, highlighting open issues and reflecting on a feedback approach that players could use to design conjectures while the coordinator has access only to its objective function. This would remove any knowledge the coordinator may have about the players' objective functions. A second direction for future extension is to link our theoretical framework with multi-agent learning algorithm, as a way to provide a theoretical base for opponent modeling, while also tackling the issue of coordinating players, improving convergence, and obtaining stability guarantees. This application could resemble the centralized-training decentralized-execution paradigm, where the centralized part would be replaced by the conjectures design problem.

\bibliography{main}
\bibliographystyle{rlc}

\end{document}